\documentclass{llncs}
\usepackage{xcolor} 
\usepackage{amsmath}
\usepackage{amssymb}
\usepackage{proof}
\usepackage{hyperref}
\usepackage{caption}
\usepackage{comment}
\usepackage{subcaption}
\usepackage{fancybox}
\usepackage{stmaryrd}
\definecolor{light-gray}{gray}{0.86}

\newcommand{\interp}[1]{\llbracket #1 \rrbracket}

\usepackage{todonotes}
\newcommand{\knote}[1]{\todo[inline, color=blue!20]{#1}}

\begin{document}
\title{A Type-theoretic Approach to Resolution \thanks{This work was funded by EPSRC grant EP/K031864/1.}}

\author{Peng Fu, Ekaterina Komendantskaya}
\institute{Computer Science, University of Dundee}

\maketitle

\begin{abstract}

  We propose a new type-theoretic approach to SLD-resolution and Horn-clause logic programming. 
  It views Horn formulas as types, and
  derivations for a given query as a construction of the inhabitant (a proof-term) for the type
  given by the query.
  We propose a method of program transformation that allows to transform logic programs in
  such a way that proof evidence is computed alongside SLD-derivations.
  We discuss two applications of this approach: in recently proposed productivity theory of 
  structural resolution, and in type class inference.
  



\end{abstract}
\section{Introduction}
\label{introduction}

Logic Programming (LP) is a programming
paradigm based on first-order Horn formulas. Informally, given a
logic program $\Phi$ and a query $A$, LP
provides a mechanism for automatically inferring whether or not $\Phi
\vdash A$ holds, i.e., whether or not
$\Phi$ logically entails $A$. The inference mechanism is based on the SLD-resolution,
which uses the resolution rule together with first-order unification.

\begin{example}\label{ex:conn}
Consider the following logic program  $\Phi$, consisting of Horn formulas labelled by $\kappa_1$, $\kappa_2$, $\kappa_3$, defining connectivity for a graph with three nodes:
 {\footnotesize
  \begin{center}
  $  \kappa_1 : \forall x. \forall y. \forall z. \mathrm{Connect}(x, y), \mathrm{Connect}(y, z) \Rightarrow \mathrm{Connect}(x, z)$

  $\kappa_2 : \ \Rightarrow \mathrm{Connect}(\mathrm{Node_1}, \mathrm{Node_2})$
  
  $\kappa_3 : \ \Rightarrow \mathrm{Connect}(\mathrm{Node_2}, \mathrm{Node_3})$
    \end{center}
}  
  
  \noindent  
  In the above program, $\mathrm{Connect}$ is a predicate, and  $\mathrm{Node_1}$ --  $\mathrm{Node_3}$ are constants.
 SLD-derivation for  the query $\mathrm{Connect}(x, y)$ can be represented as the following 
 reduction: 

 {\footnotesize 
  \begin{center}
    $\Phi \vdash \{\mathrm{Connect}(x, y)\} \leadsto_{\kappa_1, [x/x_1, y/z_1]}
    \{\mathrm{Connect}(x, y_1), \mathrm{Connect}(y_1, y)\} \leadsto_{\kappa_2, [\mathrm{Node_1}/x, \mathrm{Node_2}/y_1, \mathrm{Node_1}/x_1, y/z_1]}
    \{\mathrm{Connect}(\mathrm{Node_2}, y)\} \leadsto_{\kappa_3, [\mathrm{Node_3}/y, \mathrm{Node_1}/x, \mathrm{Node_2}/y_1,
      \mathrm{Node_1}/x_1, \mathrm{Node_3}/z_1]} \emptyset $
  \end{center}
  }
 \noindent The first reduction $\leadsto_{\kappa_1, [x/x_1, y/z_1]}$  unifies the query $\mathrm{Connect}(x, y)$ with the head of the rule $\kappa_1$ (which is $\mathrm{Connect}(x_1, z_1)$
 after renaming) with the substitution $[x/x_1, y/z_1]$ ($x_1$ is replaced by $x$ and $z_1$ is replaced by $y$). 
So the query is \emph{resolved} with $\kappa_1$, 
  producing the next queries: $\mathrm{Connect}(x, y_1)$, $\mathrm{Connect}(y_1, y)$. Note that
  the substitution in the subscript of $\leadsto$ is a state that will be updated alongside the derivation.

\end{example}

Viewing a program as a collection of  Horn clauses, the above derivation first assumed that  $\mathrm{Connect}(x, y)$ is false, and then deduced a contradiction (an empty goal) from the assumption.
As every SLD-derivation is essentially a proof by contradiction, traditionally the exact content of such proofs plays little role in determining entailment. However, it is desirable to have methods which capture the proof-theoretic content of SLD-derivations. For example, one may wish to reason in a proof-relevant way,
and compute not just $\Phi \vdash A$, but $\Phi \vdash p:A$, where $p$ is the proof-witness for the query $A$ and the program $\Phi$. LP and its dialects are used as part of type inference engines underlying functional~\cite{wadler1989make,jones2003qualified} and dependently typed \cite{gonthier2011make} languages. These applications require proof-relevant automated reasoning. 


In type class inference (e.g. Haskell), a type class can be seen as an atomic formula and an instance declaration -- as a Horn formula. The instance resolution process in type class inference can then be seen as
an SLD-derivation, with one additional requirement: this SLD-derivation must compute the   evidence for the type class (or construct a dictionary). 
For example, the following declaration specifies a way to construct equality class instances for datatypes List and Char:
 {\footnotesize
\begin{center}
\begin{tabular}{lll}
 $\kappa_1 : $ & & $\forall x .  \mathrm{Eq}(x) \Rightarrow \mathrm{Eq}(\mathrm{List} (x))$
\\
 $\kappa_2 : $  & &  $ \Rightarrow \mathrm{Eq}(\mathrm{Char})$
\end{tabular}
\end{center}
}
\noindent Here $\mathrm{List}$ is a function symbol, $\mathrm{Char}$ is a constant and $x$ is a variable; $\kappa_1, \kappa_2$ will be used as primitives for the evidence construction. When we make a comparison of two lists of characters, such as $(\mathrm{eq}\ ['a']\ ['b'])$, the compiler will insert the evidence $d$ of the type $\mathrm{Eq}(\mathrm{List} (\mathrm{Char}))$ in $(\mathrm{eq}\ d\ ['a']\ ['b'])$. The construction of this evidence can be viewed as resolving the query $\mathrm{Eq}(\mathrm{List} (\mathrm{Char}))$, which is witnessed by applying Horn formulas  $\kappa_1$ and $\kappa_2$. Thus, 
 $(\kappa_1\ \kappa_2)$ is the evidence we want for $d$.

In order to specify the proof-theoretic meaning of derivations, we introduce a type-theoretic approach to recover the notion of proof in LP.
 It has been noticed by Girard~\cite{Girard:1989}, that the resolution rule  $\frac{A \lor B \ \ \ \neg B \lor D}{A \lor D}$ can be expressed by means of 
the cut rule in  intuitionistic sequent calculus:
$\frac{A \Rightarrow B  \  \  \     B \Rightarrow D}{A \Rightarrow D}$.
Although the resolution rule is classically equivalent to the cut rule, 
the cut rule is better suited for performing computation 
 while preserving constructive content. In Section \ref{forms} we present
a type system reflecting this intuition: if $p_1$ is the proof of  $A \Rightarrow B$ and $p_2$ is the proof of $B \Rightarrow D$, then $\lambda x . p_2\ (p_1\ x)$ is the proof of $A \Rightarrow D$. Thus, proof can be recorded alongside with each cut rule. 

We prove that SLD-resolution is sound with respect to the type system (Section \ref{forms}). We give a formulation of SLD-resolution in the form
of a reduction rule, called \emph{LP-Unif}. The soundness result shows that, given a logic program $\Phi$ and a query $A$, if $A$ can be LP-Unif reduced to the empty goal with a substitution $\gamma$ as an answer, then a proof term can be constructed for $\gamma A$.

In Section \ref{real:trans}, we introduce a technique called \emph{realizability transformation}, that,
given a program $\Phi$, produces a program $F(\Phi)$ in which one extra argument is added to every predicate,
in order to record the proof-evidence in derivations. 
The proof evidence is computed by applying substitution to variables held by this additional argument in the course of running SLD-resolution. Let us revisit the List example. Its transformed version will look as follows:

 {\footnotesize
\begin{center}
\begin{tabular}{lll}
 $\kappa_1 : $ & & $\forall x . \forall u_1 . \mathrm{Eq}(x, u_1) \Rightarrow \mathrm{Eq}(\mathrm{List} (x),f_{\kappa_1}(u_1))$
\\
 $\kappa_2 : $  & &  $ \Rightarrow \mathrm{Eq}(\mathrm{Char}, c_{\kappa_2})$
\end{tabular}
\end{center}
}

\noindent The query $\mathrm{Eq}(\mathrm{List}(\mathrm{Char}))$ of the original program becomes $\mathrm{Eq}(\mathrm{List}(\mathrm{Char}), u)$ after the transformation, where $u$ is a variable.
The derivation reaches the empty goal and outputs 
the substitution $[f_{\kappa_1}(c_{\kappa_2})/u]$, which corresponds to the proof term $(\kappa_1 \ \kappa_2)$ for the query $\mathrm{Eq}(\mathrm{List}(\mathrm{Char}))$.

Realizability transformation  bears resemblance to Kleene's~\cite{KleeneSC:1952} method under the same name. 
We show that realizability transformation preserves
the proof-theoretic meaning of the original program and the computational behaviour of LP-Unif reductions. With the help of the transformation, we prove completeness of LP-Unif with
repect to the type system. 

Together, Sections \ref{forms} and \ref{real:trans} introduce a method of constructing proof evidence in the process of LP derivations.
Recently, a variant of resolution for Horn Clauses, called \emph{structural resolution} (\emph{S-resolution}) has been introduced~\cite{JKK15}.
S-resolution represents derivations by SLD-resolution as a combination of derivations by term-matching and by substitution. We explain this idea in detail in Section~\ref{rt:s}.
The main reason for separating out two components of SLD-resolution in such a way is to make use of structural properties
of term-matching that have already been exploited in functional programming and term-rewriting. In particular,  S-resolution
 allowed to define a theory of universal productivity for LP that resembles a similar theory in functional programming\cite{bertot2008inductive}:
 given a potentially infinite derivation by S-resolution, termination of term-matching derivations that comprise it determines \emph{productivity}
 of the derivation (or in other words, observability of finite fragments of the infinite computation).

 We conjecture that the combination of the two ideas -- the theory of productivity introduced by S-resolution and the proof-witness construction introduced in this paper
 bear promise for future development of resolution-based methods. This is why, in Section~\ref{rt:s} we give a full formal study of how these two methods can be formally combined.
 We show how S-resolution can be represented by means of \emph{LP-Struct reductions}, combining term-matching reductions and unification.
 We extend the type-theoretic semantics to S-Resolution.
 We define conditions which guarantee equivalence of S-Resolution and SLD-resolution, one of which happens to be exactly the property of productivity.
 We use the realizability transformation as a method for guaranteeing productivity of programs.   





 Finally, in Section~\ref{concl} we conclude and explain how the combination of S-Resolution and the type-theoretic approach of this paper could be used in non-terminating cases of type class inference. 
Detailed proofs for lemmas and theorems in this paper may be found in the extended version\footnote{Extended version is available at both authors' homepages.}. 

\section{A Type System for LP: Horn-Formulas as Types}
\label{forms}

We first formulate a type system to model LP. We show that LP-Unif is sound with respect to the type system. 

\begin{definition}

\

  Term $t \ ::= \ x \ | \ f(t_1,..., t_n)$

  Atomic Formula $A, B, C, D\ ::= \ P(t_1,...,t_n)$

  (Horn) Formula $F \ ::= \ [\forall \underline{x}]. A_1,..., A_n \Rightarrow A$


  Proof Term $p, e \ ::= \ \kappa \ | \ a \ | \ \lambda a . e \ | \ e \ e'$

  Axioms/LP Programs $\Phi \ ::= \cdot \ | \ \kappa :  F, \Phi$
\end{definition}

Functions of arity zero are called \textit{term constants}, $\mathrm{FV}(t)$ returns all free term variables of $t$.  
We use $\underline{A}$ to denote $A_1,..., A_n$, when the number $n$ is unimportant. If $n$ is zero for $\underline{A} \Rightarrow B$, then we write $\Rightarrow B$. Note that $B$ is an atomic formula, but $\Rightarrow B$ is a formula, we distinguish the notion of atomic formulas from (Horn) formulas. The formula $A_1,..., A_n \Rightarrow B$ can be informally read as ``the conjunction of $A_i$ implies $B$''. We write $\forall \underline{x} . F$ for quantifying over all the free term variables in $F$; $[\forall x].F$ denotes $F$ or $\forall x . F$. LP program $B \Leftarrow \underline{A}$ are represented as  $\forall \underline{x} . \underline{A} \Rightarrow B$ and query is an atomic formula. Proof terms are lambda terms, where $\kappa$ denotes a proof term constant and $a$ denotes a proof term variable. We write $A \mapsto_\sigma A'$ (resp. $A \sim_\gamma A'$ ) to mean $A$ is matchable (resp. unifiable) to $A'$ with substitution $\sigma$ (resp. $\gamma$), i.e. $\sigma A \equiv A'$ (resp. $\gamma A \equiv \gamma A'$).

The following is a new formulation of a type system intended to provide a type theoretic interpretation for LP.  
 
\begin{definition}[Horn-Formulas-as-Types System for LP]
  \label{proofsystem}

\
{\footnotesize
\begin{tabular}{lll}
\\
\infer[gen]{e: \forall \underline{x} . F}{e : F}
&

&
\infer[cut]{\lambda \underline{a} . \lambda \underline{b} . (e_2\ \underline{b})\ (e_1\ \underline{a}) : \underline{A}, \underline{B} \Rightarrow C}{e_1 : \underline{A} \Rightarrow D & e_2 : \underline{B}, D \Rightarrow C}
\\
\\
\infer[inst]{e : [\underline{t}/\underline{x}]F}{e : \forall \underline{x} . F}

&

&
\infer[axiom]{\kappa : \forall \underline{x}. F}{(\kappa : \forall \underline{x}. F) \in \Phi}    

  \end{tabular}
}  
\end{definition}
Note that the notion of type is identified with Horn formulas (atomic intuitionistic sequent), not atomic formulas. The usual sequent turnstile $\vdash$ is
internalized as intuitionistic implication $\Rightarrow$. The rule for first order quantification $\forall$  is placed \textit{outside} of the sequent. The cut rule is the only rule that produces new proof terms. In the \textit{cut} rule, $\lambda \underline{a}.t$ denotes $\lambda a_1....\lambda a_n.t $ and $t\ \underline{b}$ denotes $(...(t\ b_1)... b_n)$. The size of $\underline{a}$ is the same as  
$\underline{A}$ and the size of $\underline{b}$ is the same as $\underline{B}$, and $\underline{a}, \underline{b}$ are not free in $e_1, e_2$. 

Our formulation is given in the style of typed lambda calculus and sequent calculus, the intention for this formulation
is to model LP type-theoretically. It has been observed that the cut rule and proper axioms in intuitionistic sequent calculus can emulate LP~\cite{Girard:1989}(\S 13.4). Here we add a proof term annotation
and make use of explicit quantifiers. Our formulation uses Curry-style in the sense that 
for the \textit{gen} and \textit{inst} rule, we do not modify the structure of the proof terms. Curry-style
formulation allows us to focus on the proof terms generated by applying the \textit{cut} rule.

Below is a formulation of SLD-derivation as a reduction system \cite{nilsson1990logic}.

\begin{definition}[LP-Unif reduction]
\label{red}
Given axioms $\Phi$. We define a reduction relation on the multiset of atomic formulas: 

\noindent $\Phi \vdash \{A_1,..., A_i, ..., A_n\} \leadsto_{\kappa, \gamma \cdot \gamma'} \{\gamma A_1,..., \gamma B_1,..., \gamma B_m, ..., \gamma A_n\}$ for any substitution $\gamma'$, if there exists $\kappa : \forall \underline{x} . B_1,..., B_n \Rightarrow C \in \Phi$ such that $C \sim_{\gamma} A_i$.

\end{definition}
The second subscript in the reduction is intended as a state, it will be updated along with reductions. We assume implicit renaming of all quantified variables each time the above rule is applied. We write $\leadsto$ when we leave the underlining state implicit. We use $\leadsto^*$ to denote the reflexive and transitive closure of $\leadsto$. Notation $\leadsto_{\gamma}^*$ is used when the final state along the reduction path is $\gamma$. 
 

Given a program $\Phi$ and
a set of queries $\{B_1, \ldots, B_n\}$, LP-Unif uses only unification reduction to reduce $\{B_1, \ldots, B_n\}$: 

\begin{definition}[LP-Unif]
\
\noindent  Given a logic program $\Phi$, LP-Unif is given by an abstract reduction system $(\Phi, \leadsto)$. 
\end{definition}

\begin{lemma}
\label{sound}
  If $\Phi \vdash \{A_1,..., A_n\} \leadsto^*_{\gamma} \emptyset$, then there exist proofs $e_1 : \forall \underline{x} . \Rightarrow \gamma A_1,..., e_n : \forall \underline{x} . \Rightarrow \gamma A_n$, given axioms $\Phi$.
\end{lemma}

\begin{proof}
  By induction on the length of the reduction.
  
\noindent  \textit{Base Case.} Suppose the length is one, namely, $\Phi \vdash \{A\} \leadsto_{\kappa, \gamma} \emptyset$. It implies that there exists $(\kappa : \forall \underline{x} .  \Rightarrow C) \in \Phi$, such that $C \sim_\gamma A$.  So we have $\kappa :\ \Rightarrow \gamma C$ by the \textit{inst} rule. Thus $\kappa : \ \Rightarrow \gamma A$ by $\gamma C \equiv \gamma A$. Hence $\kappa : \forall \underline{x} . \Rightarrow \gamma A$ by the \textit{gen} rule.

\noindent  \textit{Step Case.} Suppose
 $\Phi \vdash \{A_1, ..., A_i,..., A_n\} \leadsto_{\kappa, \gamma} \{\gamma A_1,..., \gamma B_1,..., \gamma B_m,..., \gamma A_n\}$ $ \leadsto^*_{\gamma'} \emptyset$, where $\kappa : \forall \underline{x} . B_1,..., B_m \Rightarrow C$ and $C \sim_{\gamma} A_i$. By inductive hypothesis(IH), we know that there exist proofs $e_1 : \forall \underline{x}. \Rightarrow \gamma' \gamma A_1,..., p_1 : \forall \underline{x}. \Rightarrow \gamma' \gamma B_1,..., p_m : \forall \underline{x}. \Rightarrow \gamma' \gamma B_m,..., e_n : \forall \underline{x} . \Rightarrow \gamma' \gamma A_n$. 
We can use \textit{inst} rule to instantiate the quantifiers of $\kappa$ using $\gamma'\cdot \gamma$, so
we have $\kappa : \gamma'\gamma B_1,..., \gamma'\gamma B_m \Rightarrow \gamma'\gamma C$.
Since $\gamma'\gamma A_i \equiv \gamma' \gamma C$, we can construct a proof $e_i = \kappa \ p_1\ ...\ p_m$ with $e_i :\ \Rightarrow \gamma'\gamma A_i$, by applying the cut rule $m$ times. By \textit{gen}, we have $e_i : \forall \underline{x}. \Rightarrow \gamma'\gamma A_i$. The substitution generated by the unification is idempotent, and $\gamma'$ is accumulated from $\gamma$, i.e. $\gamma' = \gamma'' \cdot \gamma$ for some $\gamma''$,
so $\gamma' \gamma A_j \equiv \gamma'' \gamma \gamma A_j \equiv \gamma''\gamma A_j \equiv \gamma' A_j$ for any $j$. Thus we have $e_j : \forall \underline{x}. \Rightarrow \gamma' A_j$ for any $j$. 
\end{proof}


\begin{theorem}[Soundness of LP-Unif]
\label{sound:unif}
    If $\Phi \vdash \{A\} \leadsto^*_\gamma \emptyset$ , then there exists a proof $e : \forall \underline{x} .\Rightarrow \gamma A$ given axioms $\Phi$. 
\end{theorem}

For example, by the soundness theorem above, the derivation in Example \ref{ex:conn} yields the proof $(\lambda b. (\kappa_1\ b)\ \kappa_3)\ \kappa_2$ for the formula $\ \Rightarrow \mathrm{Connect}(\mathrm{node}_1, \mathrm{node}_3)$.

Naturally, we would want to prove the following completeness theorem: If $e : \forall \underline{x}. \ \Rightarrow A$, then $\Phi \vdash \{A\} \leadsto^*_\gamma \emptyset$ for some $\gamma$. It is tempting
to prove this theorem by induction on the derivation of $e : \forall \underline{x}. \ \Rightarrow A$. However, it becomes quite involved. 
We will discuss a simpler way to prove this theorem at the end of the next section, where we take advantage of the realizability transformation. 

\section{Realizability Transformation}
\label{real:trans} 
We define \textit{realizability transformation} in this section. Realizability \cite{KleeneSC:1952}(\S 82) is a technique that uses a number representing the proof of a number-theoretic formula. The transformation described here is similar in the sense that we use a first order term to represent the proof of a formula. 
More specifically, we use a first order term as an extra argument for a formula to represent a proof of that formula. Before we define the transformation, we first state several basic results about
the type system in Definition \ref{proofsystem}.

\begin{theorem}[Strong Normalization]
\label{real:sn}
Let beta-reduction on proof terms be the congruence closure of the following relation:
 $(\lambda a . p) p' \to_\beta [p'/a]p$. If $e : F$, then $e$ is strongly normalizable with respect to beta-reduction on proof terms.
\end{theorem}
 
The proof of strong normalization (SN) is an adaptation of Tait-Girard's reducibility proof. Since the first order quantification does not impact the proof term, the proof is very similar to the SN proof of simply typed lambda calculus. 
 
\begin{lemma}
\label{fst:lambda}
   If $e : [\forall \underline{x}.] \underline{A} \Rightarrow B$ given axioms $\Phi$, then either $e$ is a proof term constant or it is normalizable to the form $\lambda \underline{a}. n$, where $n$ is first order normal proof term. 
\end{lemma}
 
\begin{theorem}
  \label{fst}
    If $e : [\forall \underline{x}.] \Rightarrow B$, then $e$ is normalizable to a first order proof term.
\end{theorem}
 
Lemma \ref{fst:lambda} and Theorem \ref{fst} show that  we can use first order terms to represent normalized proof terms; and thus pave the way to realizability transformation.
 
\begin{definition}[Representing First Order Proof Terms]
\label{fst:rep}
  Let $\phi$ be a mapping from proof term variables to first order terms. We define 
a representation function $\interp{\cdot}_\phi$ from first order normal proof terms to first order terms. 
  
\noindent -- $\interp{a}_\phi = \phi(a)$. 

  \noindent -- $\interp{\kappa \ p_1 ...p_n}_\phi = f_{\kappa}(\interp{p_1}_\phi,..., \interp{p_n}_\phi)$, where $f_\kappa$ is a function symbol.

\end{definition}
 
\begin{definition}
  Let $A \equiv P(t_1,..., t_n)$ be an atomic formula, we write
  $A[t']$, where $(\bigcup_i \mathrm{FV}(t_i)) \cap \mathrm{FV}(t') = \emptyset$, to 
  abbreviate a new atomic formula $P(t_1,..., t_n, t')$.
\end{definition}
 
\begin{definition}[Realizability Transformation]
\label{real}
  We define a transformation $F$ on a formula and its normalized proof term: 
  \begin{itemize}
  \item $F(\kappa : \forall \underline{x} . A_1, ..., A_m \Rightarrow B) = \kappa : \forall \underline{x} . \forall \underline{y}. A_1[y_1], ..., A_m[y_m] \Rightarrow B[f_\kappa(y_1,...,y_m)]$, where $y_1,..., y_m$ are all fresh and distinct.
  \item $F(\lambda \underline{a} . n : [\forall \underline{x}] . A_1, ..., A_m \Rightarrow B) = \lambda \underline{a} . n : [\forall \underline{x}.\forall \underline{y}]. A_1[y_1], ..., A_m[y_m] \Rightarrow $ 

\noindent $B[\interp{n}_{[\underline{y}/\underline{a}]}]$, where $y_1,..., y_m$ are all fresh and distinct.
  \end{itemize}
     
\end{definition}

The realizability transformation systematically associates a proof to each predicate,
so that the proof can be recorded alongside with reductions. 

\begin{example}
 \label{ex:conn:real0}
The following logic program is the result of applying realizability transformation on
the program in Example \ref{ex:conn}.
  {\footnotesize
  \begin{center}
  $  \kappa_1 : \forall x . \forall y . \forall u_1. \forall u_2 . \mathrm{Connect}(x, y, u_1), \mathrm{Connect}(y, z, u_2) \Rightarrow \mathrm{Connect}(x, z, f_{\kappa_1}(u_1, u_2))$

  $\kappa_2 : \ \Rightarrow \mathrm{Connect}(\mathrm{node_1}, \mathrm{node_2}, c_{\kappa_2})$
  
  $\kappa_3 : \ \Rightarrow \mathrm{Connect}(\mathrm{node_2}, \mathrm{node_3}, c_{\kappa_3})$
    \end{center}
 }

\noindent Before the realizability transformation, we have the following judgement:
{\footnotesize
\begin{center}
  $\lambda b. (\kappa_1\ b)\ \kappa_2 :
  \mathrm{Connect}(\mathrm{node_2}, z) \Rightarrow
  \mathrm{Connect}(\mathrm{node_1}, z)$
\end{center}
}
\noindent We can apply the transformation, we get: 
{\footnotesize
\begin{center}
  $\lambda b. (\kappa_1\ b)\ \kappa_2 :
  \mathrm{Connect}(\mathrm{node_2}, z, u_1) \Rightarrow
  \mathrm{Connect}(\mathrm{node_1}, z, \interp{(\kappa_1\ b)\ \kappa_2}_{[u_1/b]})$
\end{center}
}
\noindent which is the same as
{\footnotesize
\begin{center}
  $\lambda b. (\kappa_1\ b)\ \kappa_2 :
  \mathrm{Connect}(\mathrm{node_2}, z, u_1) \Rightarrow
  \mathrm{Connect}(\mathrm{node_1}, z, f_{\kappa_1}( u_1, c_{\kappa_2}))$
\end{center}
}

\noindent Observe that the transformed formula:

\noindent $\mathrm{Connect}(\mathrm{node_2}, z, u_1) \Rightarrow \mathrm{Connect}(\mathrm{node_1}, z, f_{\kappa_1}( u_1, c_{\kappa_2}))$ is provable by $\lambda b. (\kappa_1\ b)\ \kappa_2$ using the transformed program.
\end{example}

Let $F(\Phi)$ mean applying the realizability transformation to every axiom in $\Phi$. We
write $(F(\Phi), \leadsto)$, to
mean given axioms $F(\Phi)$, use LP-Unif to reduce a given query. Note that for query $A$ in $(\Phi, \leadsto)$, it becomes query $A[t]$ for some $t$ such that $\mathrm{FV}(A) \cap \mathrm{FV}(t) = \emptyset$ in $(F(\Phi), \leadsto)$.

 The following theorem shows that realizability transformation does not change the proof-theoretic meaning of a program. This is important because it means we can apply different resolution strategies to resolve the query on the transformed program without worrying about the change of meaning. Later we will see that the behavior of LP-Struct
is different for the original program and the transformed program.

\begin{theorem}\label{th6}
\label{realI}
Given axioms $\Phi$, if $e: [\forall \underline{x}] . \underline{A}\Rightarrow B$ holds with $e$ in normal form, then $F(e : [\forall \underline{x}] . \underline{A}\Rightarrow B)$ holds for axioms $F(\Phi)$.
\end{theorem}
 
  The other direction for the theorem above is not true if we ignore the transformation $F$, namely, if $e : \forall \underline{x} . \Rightarrow A[t] $ for axioms $\Phi$, it may not be the case
that $e: \forall \underline{x} . \Rightarrow A$, since the axioms $\Phi$ may not be set up in a way
such that $t$ is a representation of proof $e$. The following theorem shows that the extra argument is used to record the term representation of the corresponding proof.
 
\begin{theorem}\label{th7}
\label{record}
 Suppose $F(\Phi) \vdash \{A[y]\} \leadsto^*_{\gamma} \emptyset$. We have $p : \forall \underline{x} . \Rightarrow \gamma A[\gamma y]$ for $F(\Phi)$, where $p$ is in normal form and $\interp{p}_{\emptyset} = \gamma y$. 
\end{theorem} 

Now we are able to show that realizability transformation will not change the unification reduction behaviour.
 
\begin{lemma}
\label{sc:unif}
  $\Phi \vdash \{A_1,..., A_n\} \leadsto^* \emptyset$ iff $F(\Phi) \vdash \{A_1[y_1],..., A_n[y_n]\} \leadsto^* \emptyset$. 
\end{lemma} 
\begin{proof}
  For each direction, by induction on the length of the reduction. Each proof will be similar to the proof of Lemma \ref{sound}. 
\end{proof} 
\begin{theorem}\label{th8}
\label{preservation}
  $\Phi \vdash \{A\} \leadsto^* \emptyset$ iff $F(\Phi) \vdash \{A[y]\} \leadsto^* \emptyset$. 
\end{theorem}

\begin{example}
 \label{ex:conn:real}
Consider the logic program after realizability transformation in Example \ref{ex:conn:real0}. Realizability transformation does not change the behaviour of LP-Unif, we still have the 
  following successful unification reduction path for query $\mathrm{Connect}(x, y, u)$:
   
  \begin{center}
{\footnotesize
$F(\Phi) \vdash \{\mathrm{Connect}(x, y, u)\}\leadsto_{\kappa_1, [x/x_1, y/z_1, f_{\kappa_1}(u_3, u_4)/u]} \{\mathrm{Connect}(x, y_1, u_3), \mathrm{Connect}(y_1, y, u_4)\}$

$\leadsto_{\kappa_2, [c_{\kappa_2}/u_3,\mathrm{node_1}/x, \mathrm{node_2}/y_1, \mathrm{node_1}/x_1, b/z_1, f_{\kappa_1}(c_{\kappa_2}, u_4)/u]} $

$\{\mathrm{Connect}(\mathrm{node_2}, y, u_4)\}$

$\leadsto_{\kappa_3, [c_{\kappa_3}/u_4, c_{\kappa_2}/u_3, \mathrm{node_3}/y, \mathrm{node_1}/x, \mathrm{node_2}/y_1,\mathrm{node_1}/x_1, \mathrm{node_3}/z_1, f_{\kappa_1}(c_{\kappa_2}, c_{\kappa_3})/u]} \emptyset $
}
\end{center}
 
\end{example}

Now let us come back to the completeness theorem. The following lemma shows that completeness result holds for the transformed program. 

\begin{lemma}\label{lm:complete}
  For $F(\Phi)$, if $n : \ \Rightarrow A[\interp{n}_\emptyset]$ where $n$ is in normal form, then $F(\Phi) \vdash \{A[\interp{n}_\emptyset]\} \leadsto^* \emptyset$.
\end{lemma}
\begin{proof}
  By induction on the structure of $n$.
  \begin{itemize}
  \item \textbf{Base Case:} $n = \kappa$. In this case, $\interp{n}_\emptyset = f_\kappa$, 
$ \kappa : \forall \underline{x}. \ \Rightarrow A'[f_\kappa] \in F(\Phi)$ and $\gamma (A'[f_\kappa]) \equiv A[f_\kappa]$ for some substitution $\gamma$. Thus $A'[f_\kappa] \sim_\gamma A[f_\kappa]$, which implies $F(\Phi) \vdash \{A[f_\kappa]\} \leadsto_{\kappa, \gamma} \emptyset$.
  \item \textbf{Step Case:} $n = \kappa\ n_1\ n_2\ ...\ n_m$. In this case, $\interp{n}_\emptyset = f_\kappa(\interp{n_1}_\emptyset, ..., \interp{n_m}_\emptyset)$, 
$\kappa : \forall \underline{x} \underline{y}. \ C_1[y_1],..., C_m[y_m]\Rightarrow B[f_\kappa(y_1,..., y_m)] \in F(\Phi)$. To obtain $n : \ \Rightarrow A[\interp{n}_\emptyset]$, we have to use 
$\kappa : \forall \underline{x}. \ C_1[\interp{n_1}_\emptyset],..., C_m[\interp{n_m}_\emptyset]\Rightarrow B[f_\kappa(\interp{n_1}_\emptyset,..., \interp{n_m}_\emptyset)]$ with $\gamma (B[f_\kappa(\interp{n_1}_\emptyset,..., \interp{n_m}_\emptyset)]) \equiv A[\interp{n}_\emptyset]$. By the inst rule, we have $\kappa : \gamma C_1[\interp{n_1}_\emptyset],..., \gamma C_m[\interp{n_m}_\emptyset]\Rightarrow \gamma B[f_\kappa(\interp{n_1}_\emptyset,..., \interp{n_m}_\emptyset)]$. Furthermore, it has to be the case that $n_1 :\ \Rightarrow \gamma C_1[\interp{n_1}_\emptyset], ..., n_m : \ \Rightarrow \gamma C_m[\interp{n_m}_\emptyset]$. Thus we have $F(\Phi) \vdash \{A[\interp{n}_\emptyset]\} \leadsto_{\kappa, \gamma} \{\gamma C_1[\interp{n_1}_\emptyset],..., \gamma C_m[\interp{n_m}_\emptyset] \}$. By IH, 
we have $F(\Phi) \vdash \{\gamma C_1[\interp{n_1}_\emptyset]\} \leadsto^*_{\gamma_1} \emptyset$. So $F(\Phi) \vdash \{A[\interp{n}_\emptyset]\} \leadsto_{\kappa, \gamma} \cdot \leadsto^*_{\gamma_1}$

\noindent $\{\gamma_1 \gamma C_2[\interp{n_2}_\emptyset],..., \gamma_1 \gamma C_m[\interp{n_m}_\emptyset] $. Again, we have $n_2 :\ \Rightarrow \gamma_1 \gamma C_2[\interp{n_2}_\emptyset], ..., n_m : \ \Rightarrow \gamma_1 \gamma C_m[\interp{n_m}_\emptyset]$. By applying IH repeatedly, we obtain $F(\Phi) \vdash \{A[\interp{n}_\emptyset]\} \leadsto^* \emptyset$.
  \end{itemize}
\end{proof}

\begin{lemma}\label{drop}
  For $F(\Phi)$, if $F(\Phi) \vdash \{A_1[t_1], ..., A_n[t_n]\} \leadsto^* \emptyset$ with $\mathrm{FV}(t_i) = \emptyset$ for all $i$, then $\Phi \vdash \{A_1,..., A_n\} \leadsto^* \emptyset$.
 \end{lemma}
 \begin{proof}
   By induction on the length of $\leadsto^*$.
   \begin{itemize}
   \item \textbf{Base Case:} $F(\Phi) \vdash \{A[f_\kappa]\} \leadsto^* \emptyset$. We have $\kappa :\ \Rightarrow A' \in \Phi$ such that $A' \sim_\gamma A$. Thus $\Phi \vdash \{A\} \leadsto_{\kappa} \emptyset$.
   \item \textbf{Step Case:} $F(\Phi) \vdash \{A_1[t_1], ..., A_i[t_i] ,..., A_n[t_n]\} \leadsto_{\kappa, \gamma} $
     
     \noindent $\{\gamma A_1[t_1], ..., \gamma B_1[t'_1], ..., \gamma B_l[t'_l],..., \gamma A_n[t_n]\} \leadsto^* \emptyset$ with $t_i \equiv f_\kappa(t'_1,..., t'_l)$ and
$\kappa : B_1,..., B_l \Rightarrow C \in \Phi$ where $C \sim_\gamma A_i$. So by IH, we have
$\Phi \vdash \{A_1, ..., A_n\} \leadsto_{\kappa, \gamma} \{\gamma A_1, ..., \gamma B_1, ..., \gamma B_l,..., \gamma A_n\} \leadsto^* \emptyset$.
   \end{itemize}
 \end{proof}

Now we are ready to prove the completeness result. 
\begin{theorem}[Completeness]
   \label{thm:complete}
  If $n : [\forall \underline{x}]. \ \Rightarrow A$, where $n$ is in normal form, then $\Phi \vdash \{A\} \leadsto^*_\gamma \emptyset $.
\end{theorem}
\begin{proof}
By Theorem \ref{th6}, we have $n : [\forall \underline{x}]. \ \Rightarrow A[\interp{n}_\emptyset]$ holds in $F(\Phi)$. By Lemma \ref{lm:complete}, we have $F(\Phi) \vdash \{A[\interp{n}_\emptyset]\} \leadsto^* \emptyset$. By Lemma \ref{drop}, we have $\Phi \vdash \{A\} \leadsto^*_\gamma \emptyset $.
\end{proof}

The completeness result relies on realizability transformation to record the proof steps for a query, so the LP-Unif reduction can just follow the proof steps to reduce the query to the empty set.
Together with Theorem \ref{sound:unif}, this proof system gives new semantics for derivations in LP.

\section{Structual Resolution}
\label{rt:s}

S-resolution~\cite{JKK15} is a newly proposed alternative to SLD-resolution that 
allows a systematic separation of derivations into term-matching and unification steps. 
A logic program is called \textit{productive} if the term-matching reduction is terminating for any query. For productive programs with coinductive meaning, finite term-rewriting reductions can be seen as measures of observation in an infinite derivation. The ability to handle corecursion in a productive way is an attractive computational feature of S-resolution.


\begin{example}\label{ex:str}
  The following program defines the predicate $\mathrm{Stream}$:
 {\footnotesize  
\begin{center}
$\kappa_1 : \forall x . \forall y . \mathrm{Stream}(y) \Rightarrow \mathrm{Stream}(\mathrm{Cons}( x, y))$  
\end{center}
}

\noindent It will result in infinite LP-Unif reduction: 

 {\footnotesize
  \begin{center}
    $\Phi \vdash \{\mathrm{Stream}(\mathrm{Cons}(x,y))\} \leadsto_{\kappa_1, [x/x_1, y/y_1]}
    \{\mathrm{Stream}(y)\} \leadsto_{\kappa_1, [\mathrm{Cons}(x_2, y_2)/y]}
    \{\mathrm{Stream}(y_2)\} \leadsto_{\kappa_1,  [\mathrm{Cons}(x_3, y_3)/y_2]} \ldots$
  \end{center}
}

\noindent But it will yield finite term-matching reduction since $\mathrm{Stream}(y)$ can not be matched
by the head of $\kappa_1$ ($\mathrm{Stream}(\mathrm{Cons}(x,y))$): 

 {\footnotesize 
  \begin{center}
    $\Phi \vdash \{\mathrm{Stream}(\mathrm{Cons}(x,y))\} \to_{\kappa_1}
    \{\mathrm{Stream}(y)\} \not \to$
  \end{center}
}

\end{example}

In general, term-matching reductions are not complete relative to LP-Unif reductions, 
but we can combine them with substitutional  steps to complete derivations. This is exactly the idea behind S-resolution.

\begin{example}\label{ex:bl}
The following program defines bits and lists of bits:
 {\footnotesize
 \begin{center}
     $\kappa_1 :\  \Rightarrow \mathrm{Bit}(0)$\\
   $\kappa_2 :\ \Rightarrow \mathrm{Bit}(1)$\\
    $\kappa_3 :\ \Rightarrow \mathrm{BList}(\mathrm{Nil})$\\
  $\kappa_4 : \forall x . \forall y.  \mathrm{BList}(y), \mathrm{Bit}(x) \Rightarrow \mathrm{BList}(\mathrm{Cons}( x, y))$\\

\end{center}} 

\noindent LP-Unif would give a complete reduction:
 
 {\footnotesize
  \begin{center} 
    $\Phi \vdash \{\mathrm{BList}(\mathrm{Cons}(x,y))\} \leadsto_{\kappa_4, [x/x_1,y/y_1]}
    \{\mathrm{Bit}(x),  \mathrm{BList}(y)\} \leadsto_{\kappa_1, [0/x, 0/x_1,y/y_1]} 
     \{\mathrm{BList}(y)\} \leadsto_{\kappa_3, [\mathrm{Nil}/y, 0/x, 0/x_1,\mathrm{Nil}/y_1]}  \emptyset$ 
  \end{center} } 

\noindent But term-matching reduction will not be able to compute an answer in this case.  
 
 {\footnotesize
  \begin{center} 
    $\Phi \vdash \{\mathrm{BList}(\mathrm{Cons}(x,y))\} \to_{\kappa_4}
    \{ \mathrm{Bit}(x),  \mathrm{BList}(y)\} \not \to $ 
  \end{center} } 

 \noindent This is why, S-resolution combines term-matching reductions with additional substitutional steps, in order  to compute the same answer:
 
 {\footnotesize
  \begin{center} 
    $\Phi \vdash \{\mathrm{BList}(\mathrm{Cons}(x,y))\} \to_{\kappa_4}
    \{ \mathrm{Bit}(x),  \mathrm{BList}(y)\} \hookrightarrow_{\kappa_1, [0/x]} \{ \mathrm{Bit}(0),  \mathrm{BList}(y)\} \to_{\kappa_1, [0/x]} \{ \mathrm{BList}(y)\} \hookrightarrow_{\kappa_3, [0/x, \mathrm{Nil}/y]} \{ \mathrm{BList}(\mathrm{Nil})\} \to_{\kappa_3, [0/x, \mathrm{Nil}/y]} \emptyset$ 
  \end{center} } 

\end{example}

 Completing derivation for Stream in the same way will result in an infinite derivation, in which every term-matching reduction is finite.

In this section, we embed S-resolution into the type theoretic framework we have developed in the previous sections.
We first define S-derivations in terms of LP-Struct reductions, in the uniform style with LP-Unif reductions, thereby also defining LP-TM reductions, which resemble reductions in term-rewriting systems~\cite{bezem2003term}. We then prove that LP-Unif and LP-Struct are operationally equivalent subject to two conditions: productivity and non-overlapping.
Finally, we show how realizability transformation can be used to guarantee productivity of logic programs in the setting of S-resolution.


\subsection{S-resolution in the Type-Theoretic Setting}
 
\begin{definition}
\

  \begin{itemize}
  \item \textbf{Term-matching(LP-TM) reduction:}

$\Phi \vdash \{A_1,..., A_i, ..., A_n\} \to_{\kappa, \gamma'} \{A_1,..., \sigma B_1,..., \sigma B_m, ..., A_n\}$ for any substitution $\gamma'$, if there exists $\kappa : \forall \underline{x} . B_1,..., B_n \Rightarrow C \in \Phi$ such that $C \mapsto_{\sigma} A_i$.

\item \textbf{Substitutional reduction:} 

$\Phi \vdash \{A_1,..., A_i, ..., A_n\} \hookrightarrow_{\kappa, \gamma \cdot \gamma'} \{\gamma A_1,..., \gamma A_i, ..., \gamma A_n\}$ for any substitution $ \gamma'$, if there exists $\kappa : \forall \underline{x} . B_1,..., B_n \Rightarrow C \in \Phi$ such that $C \sim_{\gamma} A_i$.

\end{itemize}

\end{definition}

The second subscript of term-matching reduction is used to store the substitutions obtained by unification, it is only used when we combine term-matching reductions with substitutional reductions. The second subscript in 
the substitutional reduction is intended as a state, it will be updated along with reductions. 

Given a program $\Phi$ and
a set of queries $\{B_1, \ldots, B_n\}$, LP-TM uses only term-matching reduction to reduce $\{B_1, \ldots, B_n\}$: 

\begin{definition}[LP-TM]
\
\noindent Given a logic program $\Phi$,  LP-TM is given by an abstract reduction system $(\Phi, \to)$. 
\end{definition}

LP-TM is also sound w.r.t. the type system
of Definition \ref{proofsystem}, which implies that we can obtain a proof for each successful query. 

\begin{theorem}[Soundness of LP-TM]
\label{sound:tm}
    If $\Phi \vdash \{A\} \to^* \emptyset$ , then there exists a proof $e : \forall \underline{x} .\Rightarrow  A$ given axioms $\Phi$.
\end{theorem} 


Comparing Theorem \ref{sound:unif} and Theorem \ref{sound:tm}, we see that for LP-TM, there is no need to accumulate substitutions, and the resulting formula is proven as stated. This difference is due to the use of term-matching instead of unification for the reduction. The following example 
shows that the LP-TM is incomplete with respect to the type system. 

\begin{example}
Consider the following program $\Phi$.

\begin{center}
  $\kappa_1 :\ \Rightarrow Q(\mathrm{C})$

  $\kappa_2 :\ \forall x . \forall y. Q(x) \Rightarrow P(y)$
\end{center}

\noindent For query $P(\mathrm{C})$, we have $\Phi \vdash \{P(\mathrm{C})\} \to_{\kappa_2} \{Q(x)\} \not \to$. However, there exist a proof $(\kappa_2 \ \kappa_1) :\ \Rightarrow P(\mathrm{C})$, by instantiating $x,y$ to $\mathrm{C}$ in $\kappa_2$.
\end{example}

  We use $\to^\mu$ to denote a reduction path to a $\to$-normal form. If the
  $\to$-normal form does not exist, then $\to^\mu$ denotes an infinite reduction path. We write $\hookrightarrow^1$ to denote at most one step of $\hookrightarrow$.



We can now formally define S-Resolution within our formal framework. Given a program $\Phi$ and
a set of queries $\{B_1, \ldots, B_n\}$, LP-Struct  first uses term-matching reduction to reduce  $\{B_1, \ldots, B_n\}$ to a normal form, then performs one step substitutional reduction, and then repeats this process.

\begin{definition}[Structural Resolution (LP-Struct)]
\
\noindent  Given a logic program $\Phi$,  LP-Struct is given by an abstract reduction system $(\Phi, \to^{\mu} \cdot \hookrightarrow^1)$.

\end{definition}

If a finite term-matching reduction path does not exist, then $\to^{\mu} \cdot \hookrightarrow^1$ denotes an infinite path. 
When we write $\Phi \vdash \{\underline{A}\} (\to^{\mu} \cdot \hookrightarrow^1)^* \{\underline{C}\}$, it means a nontrivial finite path will be of the shape $\Phi \vdash \{\underline{A}\} \to^{\mu} \cdot \hookrightarrow \cdot ... \cdot \to^{\mu} \cdot \hookrightarrow \cdot \to^\mu \{\underline{C}\}$. 

Now let us see the execution trace of Stream using LP-Struct: 

 {\footnotesize 
  \begin{center}
    $\Phi \vdash \{\mathrm{Stream}(\mathrm{Cons}(x,y))\} \to_{\kappa_1}
    \{\mathrm{Stream}(y)\}  \hookrightarrow_{\kappa_1, [\mathrm{Cons}(x_2, y_2)/y]} \{\mathrm{Stream}(\mathrm{Cons}(x_2, y_2))\} \to_{\kappa_1, [\mathrm{Cons}(x_2, y_2)/y]} 
    \{\mathrm{Stream}(y_2)\} \hookrightarrow_{\kappa_1,  [\mathrm{Cons}(x_3, y_3)/y_2, \mathrm{Cons}(x_2,\mathrm{Cons}(x_3, y_3) )/y]} \{\mathrm{Stream}(\mathrm{Cons}(x_3, y_3))\} \to_{\kappa_1,  [\mathrm{Cons}(x_3, y_3)/y_2, \mathrm{Cons}(x_2,\mathrm{Cons}(x_3, y_3) )/y]} \{\mathrm{Stream}(y_3)\} \ldots $
  \end{center}
}

\noindent Note that the overall reduction is infinite, but each LP-TM reduction is finite.

\subsection{LP-Struct and LP-Unif}

The next question one may ask is how LP-Struct compares to LP-Unif. They are not equivalent.
Consider the program and the finite LP-Unif derivation of Example~\ref{ex:conn}. LP-Unif has
 a finite successful derivation for the query $\mathrm{Connect}(x, y)$, but we have the following non-terminating reduction by LP-Struct:  
 {\footnotesize

\begin{center} 
  $\Phi \vdash \{\mathrm{Connect}(x, y)\} \to_{\kappa_1} \{\mathrm{Connect}(x, y_1), \mathrm{Connect}(y_1, y)\} $

$\to_{\kappa_1} \{\mathrm{Connect}(x, y_2), \mathrm{Connect}(y_2, y_1), \mathrm{Connect}(y_1, y)\} \to_{\kappa_1} ... $ 
\end{center}
}

 The diverging behavior above is due to the divergence of LP-TM reduction.
Therefore, the program of Example~\ref{ex:conn} is not productive
in the sense of~\cite{komendantskaya2014,JKK15}. 

\begin{definition}[Productivity]
  We say a program $\Phi$ is productive iff every $\to$-reduction is finite.
\end{definition}

Perhaps LP-Unif and LP-Struct are operationally equivalent for all productive programs?
The following example shows this is not the case.

\begin{example}
\label{overlap}
\
{\footnotesize
  \begin{center}

    \noindent $\kappa_1 :\ \Rightarrow P(\mathrm{C})$ 

   \noindent  $\kappa_2 :  \forall x . Q(x) \Rightarrow P(x)$
  \end{center}
}
\noindent Here $\mathrm{C}$ is a constant. The program is $\to$-terminating. However, for query $P(x)$, we have $\Phi \vdash \{P(x)\} \leadsto_{\kappa_1, [\mathrm{C}/x]} \emptyset$ with LP-Unif, but $\Phi \vdash \{P(x)\} \to_{\kappa_2} \{Q(x)\} \not \hookrightarrow$ for LP-Struct.
\end{example}

Thus, productivity is insufficient for establishing the relation between LP-Struct and LP-Unif. In Example \ref{overlap}, the problem is caused by
the overlapping heads $P(\mathrm{C})$ and $P(x)$. Motivated by the notion of non-overlapping rules  in term rewriting systems (\cite{bezem2003term,Baader:1998}), we introduce the following definition.

\begin{definition}[Non-overlapping Condition]
  Axioms $\Phi$ are non-overlapping if for any two formulas $\forall \underline{x}. \underline{B} \Rightarrow C, \forall \underline{x}. \underline{D} \Rightarrow E \in \Phi$, there are no substitution $\sigma, \delta$ such that $\sigma C \equiv \delta E$.
\end{definition}


\begin{theorem}
  \label{ortho:equiv}
  Suppose $\Phi$ is non-overlapping. $\Phi \vdash \{A_1, ..., A_n\} \leadsto^*_\gamma \{C_1, ..., C_m\}$ with $\{C_1, ..., C_m\}$ in $\leadsto$-normal form iff $\Phi \vdash \{A_1, ..., A_n\} (\to^\mu \cdot \hookrightarrow^1)^*_\gamma \{C_1, ..., C_m\}$ with $\{C_1, ..., C_m\}$ in $\to^\mu \cdot \hookrightarrow^1$-normal form.   
\end{theorem}

The theorem above still requires the termination of the $\leadsto$ to establish equivalence LP-Unif and LP-Struct. We can weaken this requirement by only requiring
termination of the $\to$-reduction, i.e. by requiring productivity. 

\begin{theorem}[Equivalence of LP-Struct and LP-Unif]\label{prod-non-overlap}
  \
  Suppose $\Phi$ is non-overlapping and productive. 
  \begin{enumerate}
  \item If $\Phi \vdash \{A_1,..., A_n\} \leadsto \{B_1,..., B_m\}$, then $\Phi \vdash \{A_1,..., A_n\} (\to^\mu \cdot \hookrightarrow^1)^* \{C_1,..., C_l\}$ and $\Phi \vdash \{B_1,..., B_m\} \to^* \{C_1,..., C_l\}$.
   \item If $\Phi \vdash \{A_1,..., A_n\} (\to^\mu \cdot \hookrightarrow^1)^* \{B_1,..., B_m\}$, then $\Phi \vdash \{A_1,..., A_n\} \leadsto^* \{B_1,..., B_m\}$.  
  \end{enumerate}
\end{theorem}

Note that the above theorem does not rely on termination of LP-Unif reductions and therefore establishes equivalence of LP-Unif and LP-Struct even for coinductive programs 
like
Stream of Example \ref{ex:str}, as long
as they are productive and non-overlapping. 
This effect of productivity has not been described in previous work.

\subsection{Realizability Transformation and LP-Struct}

Even when programs 
 are overlapping and unproductive (as e.g. the program of Example \ref{ex:conn}),  we would still like
 to obtain a meaningful execution behaviour for LP-Struct, especially if LP-Unif allows successful derivations for the programs.
Luckily, we already have a method to achieve that, it is the realizability transformation defined in Section~\ref{real:trans}:  

\begin{proposition}
  \label{trans:equiv}
 For any program $\Phi$, $F(\Phi)$ is productive and non-overlapping.
\end{proposition} 
\begin{proof}
First, we need to show $\to$-reduction is strongly normalizing in $(F(\Phi), \to)$. By Definition \ref{real}, we can establish a decreasing measurement(from right to left, using the strict subterm relation) for each rule in $F(\Phi)$, since the last argument in the head of each rule is strictly larger than the ones in the body. Then, non-overlapping property is due to the fact that all the heads of the rules in $F(\Phi)$ will be \textit{guarded} by the unique function symbol in Definition \ref{real}.
\end{proof}


\begin{corollary}
\

 $F(\Phi) \vdash \{A_1,..., A_n\} (\to^\mu \cdot \hookrightarrow^1)^* \{B_1,..., B_m\}$ iff $F(\Phi) \vdash \{A_1,..., A_n\} \leadsto^* \{B_1,..., B_m\}$.  
\end{corollary}
\begin{proof}
  By Theorem \ref{prod-non-overlap} and Theorem \ref{trans:equiv}. 
\end{proof}
\begin{example}
\label{ex:conn:real1}
For the program in Example \ref{ex:conn:real0}, the query $\mathrm{Connect}(x, y, u)$ can be reduced by LP-Struct successfully:   
 
  \begin{center}
{\footnotesize  
$F(\Phi) \vdash \{\mathrm{Connect}(x, y, u)\} \hookrightarrow_{\kappa_1, [x/x_1, y/z_1, f_{\kappa_1}(u_3, u_4)/u]} \{\mathrm{Connect}(x, y,f_{\kappa_1}(u_3, u_4))\} \to_{\kappa_1} \{\mathrm{Connect}(x, y_1, u_3), \mathrm{Connect}(y_1, y, u_4)\}$

$\hookrightarrow_{\kappa_2, [c_{\kappa_2}/u_3,\mathrm{node_1}/x, \mathrm{node_2}/y_1, \mathrm{node_1}/x_1, b/z_1, f_{\kappa_1}(c_{\kappa_2}, u_4)/u]} \{\mathrm{Connect}(\mathrm{node_1},\mathrm{node_2}, c_{\kappa_2}), \mathrm{Connect}(\mathrm{node_2}, y, u_4)\} \to_{\kappa_2} \{\mathrm{Connect}(\mathrm{node_2}, y, u_4)\}$

$\hookrightarrow_{\kappa_3, [c_{\kappa_3}/u_4, c_{\kappa_2}/u_3, \mathrm{node_3}/y, \mathrm{node_1}/x, \mathrm{node_2}/y_1,\mathrm{node_1}/x_1, \mathrm{node_3}/z_1, f_{\kappa_1}(c_{\kappa_2},c_{\kappa_3})/u]}  \{\mathrm{Connect}(\mathrm{node_2}, \mathrm{node_3}, c_{\kappa_3})\}  \to_{\kappa_3} \emptyset $}
\end{center}
   
\noindent Note that the answer for $u$ is $f_{\kappa_1}(c_{\kappa_2},c_{\kappa_3})$, which is the first order term representation of the proof of $ \ \Rightarrow \mathrm{Connect}(\mathrm{node}_1, \mathrm{node}_3)$. 
\end{example}

Realizability transformation uses the extra argument as decreasing measurement
in the program to achieve termination of $\to$-reduction.  At the same time this extra argument makes the program non-overlapping.  
Realizability 
transformation does not modify the proof-theoretic meaning and the execution behaviour of LP-Unif. 
The next example shows that not every transformation technique for obtaining structurally decreasing LP-TM reductions has such properties:

\begin{example} Consider the following program:

{\footnotesize
\begin{center}
\noindent $\kappa_1 : \ \Rightarrow P(\mathrm{Int})$  

\noindent $\kappa_2 : \forall x . P(x), P(\mathrm{List}(x)) \Rightarrow P(\mathrm{List}(x))$
\end{center} 
}
\noindent It is a folklore method to add a structurally decreasing argument as a  measurement to ensure finiteness of $\to^\mu$.

{\footnotesize
\begin{center}

\noindent $\kappa_1 : \ \Rightarrow P(\mathrm{Int}, 0)$  

\noindent $\kappa_2 : \forall x . \forall y . P(x, y), P(\mathrm{List}(x), y) \Rightarrow P(\mathrm{List}(x), \mathrm{s}(y))$
\end{center} 
} 
\noindent We denote the above program as $\Phi'$. Indeed with the measurement we add, the term-matching reduction in $\Phi'$ will be finite. But the reduction for query $P(\mathrm{List}(\mathrm{Int}), z)$ using unification will fail: 

{\footnotesize
\begin{center}
$\Phi' \vdash \{P(\mathrm{List}(\mathrm{Int}), z) \}\leadsto_{\kappa_2, [ \mathrm{Int}/x,\mathrm{s}(y_1)/z]} \{P(\mathrm{Int}, y_1), P(\mathrm{List}(\mathrm{Int}), y_1)\}\leadsto_{\kappa_2, [ 0/y_1, \mathrm{Int}/x,\mathrm{s}(0)/z]} \{P(\mathrm{List}(\mathrm{Int}), 0)\} \not \leadsto$ 
\end{center}
}
\noindent However, the query $P(\mathrm{List}(\mathrm{Int}))$ on the original program using unification reduction will diverge. Divergence and failure are operationally different. Thus
adding arbitrary measurement may modify the execution behaviour of a program (and hence the meaning of the program). In contrast, 
by Theorems~\ref{th6}-\ref{th8},
realizability transformation does not modify the execution behaviour of unification reduction.
\end{example}

\begin{example}
Consider the following non-productive and non-overlapping program and its version after the realizability transformation:

\begin{center}
\emph{Original program:}  $\kappa : \forall x . P(x) \Rightarrow P(x)$\\
\emph{After transformation:}  $\kappa : \forall x . \forall u. P(x, u) \Rightarrow P(x, f_\kappa(u))$
\end{center}

\noindent Both LP-Struct and LP-Unif will diverge for the queries $P(x), P(x, y)$ in both original and transformed versions. LP-Struct reduction diverges for different reasons in the two cases, one is due to divergence of
$\to$-reduction:

\noindent $\Phi \vdash \{P(x)\} \to \{P(x)\} \to \{P(x)\} ...$

\noindent The another is due to $\hookrightarrow$-reduction:

\noindent $\Phi \vdash \{P(x, y)\} \hookrightarrow \{P(x, f_k(u))\} \to \{P(x, u)\} \hookrightarrow \{P(x, f_k(u'))\} \to \{P(x, u')\} ...$

 Note that a single step of LP-Unif reduction for the original program corresponds to infinite steps of term-matching reduction in LP-Struct. For the transformed version, a single step of LP-Unif reduction corresponds to finite steps of LP-Struct reduction, which is exactly the correspondence we were looking for.
\end{example}

\section{Conclusions and Future Work}\label{concl}


  

We proposed a type system that gives a proof theoretic interpretation for LP: Horn formulas correspond to the notion of type, and a successful query yields a first order proof term. The type system also provided us with  a precise tool to show that realizability transformation preserves both proof-theoretic meaning of the program and the operational behaviour of LP-Unif. 


We formulated S-resolution as LP-Struct reduction, which can be seen as a reduction strategy that combines term-matching reduction with substitutional reduction.
This formulation allowed us to study the operational relation between LP-Struct and LP-Unif. 
The operational equivalence of LP-Struct and LP-Unif is by no means obvious.
Previous work (\cite{JKK15,komendantskaya2014})
only gives soundness and completeness of LP-Struct with respect to the Herbrand models.
 We identified that productivity and non-overlapping are essential for showing their operational equivalence.


Realizability transformation proposed here ensures that the resulting programs are productive and non-overlapping. 
It  preserves the proof-theoretic meaning of the program, in a formally defined sense of Theorems~\ref{th6}-\ref{th8}.
It is general, applies to any logic program, and can be easily mechanised. Finally, it allows to automatically record the proof content in the course of reductions, as Theorem~\ref{th7} establishes, which helps to prove completeness of LP-Unif (Theorem \ref{thm:complete}).

With the proof system for LP-reductions we proposed, we are planning to further investigate the interaction of LP-TM/Unif/Struct with typed functional languages. We expect to find a tight connection
between our work and the type class inference, cf.~\cite{wadler1989make,jones2003qualified}.

  In the context of type class inference~\cite{wadler1989make,jones2003qualified}, the infinite term-matching behaviour seems pervasive. The example below specifies a possible equality instance declaration for 
nested datatype such as\\ \texttt{data Bush a = Nil | Con a (Bush (Bush a))}:

 {\footnotesize
\begin{center}
\begin{tabular}{lll}
 $\kappa_1 : $ & & $ \mathrm{Eq}(x), \mathrm{Eq}(\mathrm{Bush}(\mathrm{Bush}(x))) \Rightarrow \mathrm{Eq}(\mathrm{Bush} (x))$
\\
 $\kappa_2 : $  & &  $ \Rightarrow \mathrm{Eq}(\mathrm{Char})$
\end{tabular}
\end{center}
}

\noindent Here $\mathrm{Bush}$ is a function symbol, $\mathrm{Char}$ is a constant and $x$ is variable. 
 Consider the query $\mathrm{Eq}(\mathrm{Bush}(\mathrm{Char}))$, both LP-Unif and LP-Struct will generate an infinite reduction path by repeatedly applying $\kappa_1$. Using the realizability transformation, we can obtain a well-behaved (productive) program:




 {\footnotesize
\begin{center}
\begin{tabular}{lll}
 $\kappa_1 : $ & & $ \mathrm{Eq}(x, y_1), \mathrm{Eq}(\mathrm{Bush}(\mathrm{Bush}(x)), y_2)\Rightarrow \mathrm{Eq}(\mathrm{Bush} (x), f_{\kappa_1}(y_1, y_2))$
\\
 $\kappa_2 : $  & &  $ \Rightarrow \mathrm{Eq}(\mathrm{Char}, \mathrm{c}_{\kappa_2})$

\end{tabular}
\end{center}
}

 \noindent The substitution for $u$ in the query $\mathrm{Eq}(\mathrm{Bush}(\mathrm{Char}), u)$ will be an infinite term. But we need a finite representation for such infinite term
 to construct a dictionary. Such coinductive dictionary construction is the subject of our further investigations. We would also like to investigate generalizing the type-theoretic approach from Horn formulas to implicational
intuitionistic formulas, the type system in this case will correspond to a version of simply type lambda calculus.

\bibliographystyle{plain}
 
\bibliography{tm-lp} 
\newpage
\appendix

\section{Proof of Theorem \ref{real:sn} and \ref{fst}}
We are going to prove a nontrival property about the type system that we just set up. 
The proof is a simplification of Tait-Girard's reducibility method. 

\begin{definition}[Reducibility Set]
Let $N$ denotes the set of all strong normalizing proof terms. We define reducibility set $\mathsf{RED}_{F}$ by induction on structure of $F$: 

  \begin{itemize}
  \item $p \in \mathsf{RED}_{A_1,..., A_n \Rightarrow B}$ with $n \geq 0$ iff for any $p_i \in N$, $p \ p_1 \ ... \ p_n \in N$. 
  \item $p \in \mathsf{RED}_{\forall \underline{x}. A_1,..., A_n \Rightarrow B}$ iff $p \in \mathsf{RED}_{A_1,..., A_n \Rightarrow B}$.
  \end{itemize}
\end{definition}

\begin{lemma}
  \label{irr}
  $\mathsf{RED}_{\underline{A} \Rightarrow B} = \mathsf{RED}_{\phi \underline{A} \Rightarrow \phi B}$.
\end{lemma}
\begin{lemma}
  If $p \in \mathsf{RED}_F$, then $p \in N$. 
\end{lemma}

\begin{proof}
By Induction on $F$:
\begin{itemize}      
  \item Base Case: $F$ is of the form $A_1 , ... , A_n \Rightarrow B$. By definition, $p\ p_1\ ... \ p_n \in N$ for any $p_i \in N$. Thus $p \in N$. 
    \item Step Case: $F$ is of the form $\forall \underline{x} . A_1 , ... , A_n \Rightarrow B$. $p \in \mathsf{RED}_{\forall \underline{x} . A_1 , ... , A_n \Rightarrow B}$ implies $p \in \mathsf{RED}_{A_1 , ... , A_n \Rightarrow B}$. Thus by IH, $p \in N$. 
\end{itemize}
\end{proof}

\begin{lemma}
  \label{strong}
  If $e : F$, $e \in \mathsf{RED}_{F}$. 
\end{lemma}
\begin{proof}
  By induction on derivation of $e : F$.
  \begin{itemize}
  \item Base Case:

    \begin{tabular}{l}
\infer{\kappa : \forall \underline{x} . \Rightarrow B}{}
\end{tabular}

This case $\kappa \in N$. 
  \item Base Case:

    \begin{tabular}{l}
\infer{\kappa : \forall \underline{x} . A_1,..., A_n \Rightarrow B}{}
\end{tabular}

Since $\kappa$ is a constant, thus for any $p_i \in N$, $\kappa\ p_1\ ...\ p_n \in N$. So
$\kappa \in \mathsf{RED}_{ A_1,..., A_n \Rightarrow B}$, thus $\kappa \in \mathsf{RED}_{\forall \underline{x}. A_1,..., A_n \Rightarrow B}$.
   \item Step Case: 
\

     \begin{tabular}{l}
\infer[cut]{\lambda \underline{a} . \lambda \underline{b} . (e_2\ \underline{b})\ (e_1\ \underline{a}) : \underline{A}, \underline{B} \Rightarrow C}{e_1 : \underline{A} \Rightarrow D & e_2 : \underline{B}, D \Rightarrow C}
\end{tabular}

We need to show $\lambda \underline{a} . \lambda \underline{b} . (e_2\ \underline{b})\ (e_1\ \underline{a}) \in \mathsf{RED}_{  \underline{A},  \underline{B} \Rightarrow   C}$. By IH, we know that $(e_2\ \underline{b})\ (e_1\ \underline{a}) \in N$. Let $p_1 \in N,...,p_n \in N, q_1 \in N,..., q_m \in N$. We
are going to show for any $e$ with $(\lambda \underline{a} . \lambda \underline{b} . (e_2\ \underline{b})\ (e_1\ \underline{a})) \ \underline{p}\ \underline{q} \to_\beta e$, then $e \in N$. We proceed by induction on $(\mu((e_2\ \underline{b})\ (e_1\ \underline{a})), \mu(\underline{p}), \mu(\underline{q}))$, where $\mu$ is a function to get the length of the reduction path to normal form.
\begin{itemize}
\item Base Case: $(\mu((e_2\ \underline{b})\ (e_1\ \underline{a})), \mu(\underline{p}), \mu(\underline{q})) = (0,..., 0)$. 
The only reduction possible is $(\lambda \underline{a} . \lambda \underline{b} . (e_2\ \underline{b})\ (e_1\ \underline{a})) \ \underline{p}\ \underline{q} \to_\beta (e_2\ \underline{q})\ (e_1\ \underline{p})$. We know that $(e_2\ \underline{q})\ (e_1\ \underline{p}) \in N$.
\item Step Case: There are several possible reductions, but all will decrease $(\mu((e_2\ \underline{b})\ (e_1\ \underline{a})), \mu(\underline{p}), \mu(\underline{q}))$, thus we conclude that by induction hypothesis.
\end{itemize}
So $\lambda \underline{a} . \lambda \underline{b} . (e_2\ \underline{b})\ (e_1\ \underline{a}) \in N$. 
\item Step Case: 

  \begin{tabular}{l}
\infer[inst]{e : [\underline{t}/\underline{x}]F}{e : \forall \underline{x} . F}
\end{tabular}

By IH, we konw that $e \in \mathsf{RED}_{\forall \underline{x} . F}$, so by definition we know that $e \in \mathsf{RED}_{F}$. By Lemma \ref{irr}, $e \in \mathsf{RED}_{[\underline{t}/\underline{x}]F}$.
\item Step Case: 

  \begin{tabular}{l}
\infer[gen]{e: \forall \underline{x} . F}{e : F}
\end{tabular}

By IH, we know that $e \in \mathsf{RED}_{F}$, so we know that $e \in \mathsf{RED}_{\forall \underline{x} . F}$. 
  \end{itemize}
\end{proof}

\begin{theorem}[Strong Normalization]
  If $e : F$, then $e \in N$.
\end{theorem}
\begin{proof}
  By Lemma \ref{strong}.
\end{proof}
\begin{definition}[First Orderness]
  We say $p$ is first order inductively:
  \begin{itemize}
  \item A proof term variable $a$ or proof term constant $\kappa$ is first order.
  \item if $n, n'$ are first order, then $n\ n'$ is first order.
  \end{itemize}
\end{definition}

\begin{lemma}
  \label{fo:sub}
  If $n, n'$ are first order, then $[n'/a]n$ is first order. 
\end{lemma}

\begin{lemma}
\label{FO}
  If $e : [\forall \underline{x}.] \underline{A} \Rightarrow B$, then either $e$ is a proof term constant or it is normalizable to the form $\lambda \underline{a}. n$, where $n$ is first order normal term. 
\end{lemma}
\begin{proof}
  By induction on the derivation of $e : [\forall \underline{x}.] \underline{A} \Rightarrow B$.
  \begin{itemize}
  \item Base Cases: Axioms, in this case $e$ is a proof term constant.
    \item Step Case: 
      
           \begin{tabular}{l}
\infer[cut]{\lambda \underline{a} . \lambda \underline{b} . (e_2\ \underline{b})\ (e_1\ \underline{a}) : \underline{A}, \underline{B} \Rightarrow C}{e_1 : \underline{A} \Rightarrow D & e_2 : \underline{B}, D \Rightarrow C}
\end{tabular}

           By IH, we know that $e_1 = \kappa$ or $e_1 = \lambda \underline{a}.n_1$; $e_2 = \kappa'$ or $e_2 = \lambda \underline{b} d . n_2$. We know that $e_1 \underline{a}$ will be normalizable to a first order proof term. And $e_2 \underline{b}$ will be normalized to either $\kappa' \underline{b}$ or $\lambda d . n_2$. So by Lemma \ref{fo:sub}, we conclude that 
           $\lambda \underline{a} . \lambda \underline{b} . (e_2\ \underline{b})\ (e_1\ \underline{a})$ is normalizable to $\lambda \underline{a} . \lambda \underline{b} . n$ for some
           first order normal term $n$.
           \item The other cases are straightforward.
  \end{itemize}
\end{proof}
\begin{theorem}
  If $e : [\forall \underline{x}.] \Rightarrow B$, then $e$ is normalizable to a first order proof term.
\end{theorem}
\begin{proof}
  By lemma \ref{FO}, subject reduction and strong normalization theorem. 
\end{proof}
\section{Proof of Theorem \ref{realI}}

\begin{theorem}
Given axioms $\Phi$, if $e: [\forall \underline{x}] . \underline{A}\Rightarrow B$ holds with $e$ in normal form, then $F(e : [\forall \underline{x}] . \underline{A}\Rightarrow B)$ holds for axioms $F(\Phi)$.
\end{theorem}
\begin{proof}
  By induction on the derivation of $e: [\forall \underline{x}] . \underline{A}\Rightarrow B$.
  \begin{itemize}
  \item Base Case: 

   \

    \begin{tabular}{l}
     \infer{\kappa : \forall \underline{x} . \underline{A}\Rightarrow B}{}      

    \end{tabular}
   
\

 In this case, we know that $F(\kappa : \forall \underline{x} . \underline{A}\Rightarrow B) = \kappa : \forall \underline{x}. \forall \underline{y} . A_1[y_1], ..., A_n[y_n] \Rightarrow B[f_\kappa (y_1,..., y_n)] \in F(\Phi)$. 

\item Step Case: 

\

\begin{tabular}{l}
\infer[cut]{\lambda \underline{a} . \lambda \underline{b} . (e_2\ \underline{b})\ (e_1\ \underline{a}) : \underline{A}, \underline{B} \Rightarrow C}{e_1 : \underline{A} \Rightarrow D & e_2 : \underline{B}, D \Rightarrow C}
  
\end{tabular}

\

We know that the normal form of $e_1$ must be $\kappa_1$ or $\lambda \underline{a}. n_1$; the normal form of $e_1$ must be $\kappa_2$ or $\lambda \underline{b} d. n_2$, with $n_1, n_2$ are first order. 
\begin{itemize}
\item $e_1 \equiv \kappa_1, e_2 \equiv \kappa_2$. By IH, we know that $F(\kappa_1 : \underline{A} \Rightarrow D) = \kappa_1 : A_1[y_1],..., A_1[y_1] \Rightarrow D[f_{\kappa_1}(y_1,..., y_n)] $ and 
$F(\kappa_2 : \underline{B}, D \Rightarrow C) = \kappa_2 : B_1[z_1],..., B_m[z_m], D[y] \Rightarrow C[f_{\kappa_2}(z_1,..., z_m, y)]$ hold. So by \textit{gen} and \textit{inst}, we have

\noindent $\kappa_2 : B_1[z_1],..., B_m[z_m], D[f_{\kappa_1}(y_1,..., y_n)] \Rightarrow C[f_{\kappa_2}(\underline{z}, f_{\kappa_1}(\underline{y}))]$. 

\noindent Then by the cut rule, we have 

\noindent $\lambda \underline{a} . \lambda \underline{b} . \kappa_2 \underline{b} (\kappa_1 \underline{a}) : A_1[y_1],..., A_1[y_1],  B_1[z_1],..., B_m[z_m] \Rightarrow C[f_{\kappa_2}(\underline{z}, f_{\kappa_1}(\underline{y}))]$. We 
can see that $\interp{\kappa_2 \underline{b} (\kappa_1 \underline{a})}_{[\underline{y}/\underline{a}, \underline{z}/\underline{b}]} = f_{\kappa_2}(\underline{z}, f_{\kappa_1}(\underline{y}))$.

\item $e_1 \equiv \lambda \underline{a}. n_1, e_2 \equiv \lambda \underline{b}d. n_2$. By IH, we know that $F(\lambda \underline{a}. n_1 : \underline{A} \Rightarrow D) = \lambda \underline{a}. n_1 : A_1[y_1],..., A_1[y_1] \Rightarrow D[\interp{n_1}_{[\underline{y}/\underline{a}]}]$ and 
$F(\lambda \underline{b} d. n_2 : \underline{B}, D \Rightarrow C) = \lambda \underline{b} d. n_2 : B_1[z_1],..., B_m[z_m], D[y] \Rightarrow C[\interp{n_2}_{[\underline{z}/ \underline{b}, y/d]}]$ hold. So by \textit{gen} and \textit{inst}, we have

\noindent $\lambda \underline{b} d. n_2 : B_1[z_1],..., B_m[z_m], D[\interp{n_1}_{[\underline{y}/\underline{a}]}] \Rightarrow C[\interp{n_2}_{[\underline{z}/\underline{b}, \interp{n_1}_{[\underline{y}/\underline{a}]}/d]}]$. 

\noindent Then by the cut rule and beta reductions, we have $\lambda \underline{a} . \lambda \underline{b} . ([n_1/d]n_2) : A_1[y_1],..., A_1[y_1],  B_1[z_1],..., B_m[z_m] \Rightarrow C[\interp{n_2}_{[\underline{z}/\underline{b}, \interp{n_1}_{[\underline{y}/\underline{a}]}/d]}]$. We 
know that $\interp{[n_1/d]n_2}_{[\underline{y}/\underline{a}, \underline{z}/\underline{b}]} = \interp{n_2}_{[\underline{z}/\underline{b}, \interp{n_1}_{[\underline{y}/\underline{a}]}/d]}$. 
\item The other cases are handle similarly. 
\end{itemize}

\item Step Case:

\

  \begin{tabular}{l}
\infer[inst]{\lambda \underline{a} . n : [\underline{t}/\underline{x}]\underline{A} \Rightarrow [\underline{t}/\underline{x}]B }{\lambda \underline{a} . n : \forall \underline{x} . \underline{A} \Rightarrow B}
\end{tabular}

\

By IH, we know that $F(\lambda \underline{a} . n : \forall \underline{x} . \underline{A} \Rightarrow B) = \lambda \underline{a} . n : \forall \underline{x}. \forall \underline{y} . A_1[y_1],..., A_n[y_n] \Rightarrow B[\interp{n}_{[\underline{y}/\underline{a}]}]$ holds for $F(\Phi)$. By \textit{Inst} rule, we instantiate $y_i$ with $y_i$,  we have 
$\lambda \underline{a} . n : [\underline{t}/\underline{x}] A_1[y_1],..., [\underline{t}/\underline{x}] A_n[y_n] \Rightarrow [\underline{t}/\underline{x}] B[\interp{n}_{[\underline{y}/\underline{a}]}]$

\item Step Case:

\

  \begin{tabular}{l}
\infer[gen]{e: \forall \underline{x} . F}{e : F}
\end{tabular}

\

This case is straightforwardly by IH. 
  \end{itemize}
\end{proof}

\section{Proof of Theorem \ref{record}}

\begin{lemma}
\label{realII}
  If $F(\Phi) \vdash \{A_1[y_1],..., A_n[y_n]\} \leadsto^*_{\gamma} \emptyset$, and $y_1,..., y_n$ are fresh, then there exists proofs $e_1 : \forall \underline{x} . \Rightarrow \gamma A_1[\gamma y_1],..., e_n : \forall \underline{x} . \Rightarrow \gamma A_n[\gamma y_n]$ with $\interp{e_i}_{\emptyset} = \gamma y_i $ given axioms $F(\Phi)$.  
\end{lemma}
\begin{proof}
  By induction on the length of the reduction. 
  \begin{itemize}
  \item Base Case. Suppose the length is one, namely, $F(\Phi) \vdash \{A[y]\} \leadsto_{\kappa, \gamma_1} \emptyset$. Thus there exists $(\kappa : \forall \underline{x} .  \Rightarrow C[f_\kappa]) \in F(\Phi)$(here $f_\kappa$ is a constant), such that $C[f_\kappa] \sim_{\gamma_1} A[y]$.  Thus $ \gamma_1 (C[f_\kappa]) \equiv \gamma_1 A[\gamma_1 y]$. So $\gamma_1 y \equiv f_\kappa$
and $\gamma_1 C \equiv \gamma_1 A$. We have $\kappa :\ \Rightarrow  \gamma_1 C[f_\kappa]$ by the \textit{inst} rule, thus $\kappa :\ \Rightarrow \gamma_1 A[\gamma_1 y]$, hence $\kappa : \forall \underline{x} . \Rightarrow \gamma_1 A[\gamma_1 y]$ by the \textit{gen} rule and $\interp{\kappa}_{\emptyset} = f_{\kappa}$.

  \item Step Case. Suppose $F(\Phi) \vdash \{A_1[y_1], ..., A_i[y_i],..., A_n[y_n]\} \leadsto_{\kappa, \gamma_1} $ 

$\{\gamma_1 A_1[y_1],...,  \gamma_1 B_1[z_1],...,  \gamma_1 B_m[z_m],..., \gamma_1 A_n[y_n]\} \leadsto^*_{\gamma} \emptyset$,

\noindent  where $\kappa : \forall \underline{x} . \forall \underline{z} . B_1[z_m],..., B_n[z_m] \Rightarrow C[f_\kappa(z_1,..., z_m)] \in F(\Phi)$,

\noindent and $C[f_\kappa(z_1,..., z_m)] \sim_{ \gamma_1} A_i[y_i]$. So we know $ \gamma_1 C[f_\kappa(z_1,..., z_m)] \equiv \gamma_1 A_i[\gamma_1 y_i]$,  $\gamma_1 y_i \equiv f_\kappa(z_1,..., z_m),  \gamma_1 C \equiv \gamma_1 A_i$ and 

\noindent $\mathrm{dom}(\gamma_1) \cap \{z_1,..., z_m, y_1,..,y_{i-1}, y_{i+1}, y_n\} = \emptyset$. By IH, we know that there exists proofs $e_1 : \forall \underline{x}. \Rightarrow \gamma \gamma_1 A_1[\gamma y_1],..., p_1 : \forall \underline{x}. \Rightarrow \gamma   \gamma_1 B_1[\gamma z_1],..., p_m : \forall \underline{x}. \Rightarrow \gamma   \gamma_1 B_m[\gamma z_m],..., e_n : \forall \underline{x} . \Rightarrow \gamma \gamma_1 A_n[\gamma y_n]$ and $\interp{e_1}_{\emptyset} = \gamma y_1, ..., \interp{p_1}_\emptyset = \gamma z_1, ..., \interp{e_n}_{\emptyset} = \gamma y_n $ . We can construct a proof $e_i = \kappa \ p_1\ ... p_m$ with $e_i : \forall \underline{x} . \Rightarrow \gamma \gamma_1 A_i[\gamma \gamma_1 y_i]$, by first use the \textit{inst} to instantiate the quantifiers of $\kappa$, then applying the cut rule $m$ times. Moreover, we have $\interp{\kappa \ p_1\ ... p_m}_{\emptyset} = f_\kappa(\interp{p_1}_\emptyset,...,\interp{p_m}_\emptyset) = \gamma (f_\kappa (z_1,..., z_m)) = \gamma \gamma_1 y_i$. 
     \end{itemize}

\end{proof}

\begin{theorem}
  Given axioms $\Phi$, suppose $F(\Phi) \vdash \{A[y]\} \leadsto^*_{\gamma} \emptyset$. We have $p : \forall \underline{x} . \Rightarrow \gamma A[\gamma y]$ where $p$ is in normal form and $\interp{p}_{\emptyset} = \gamma y$. 
\end{theorem}
\begin{proof}
By Lemma \ref{realII}.
\end{proof}

\section{Proof of Lemma \ref{sc:unif}}
\begin{lemma}
If $\Phi \vdash \{A_1,..., A_n\} \leadsto^* \emptyset$, then $F(\Phi) \vdash \{A_1[y_1],..., A_n[y_n]\} \leadsto^* \emptyset$ with $y_i$ fresh. 
\end{lemma}
\begin{proof}
  By induction on the length of reduction.
  \begin{itemize}
     \item Base Case. Suppose the length is one, namely, $\Phi \vdash \{A\} \leadsto_{\kappa, \gamma_1} \emptyset$. Then there exists $(\kappa : \forall \underline{x} .\  \Rightarrow C) \in \Phi$ such that $C \sim_{\gamma_1} A$.  Thus $\kappa : \forall \underline{x}. \ \Rightarrow C[f_\kappa] \in F(\Phi)$ and $(C[f_\kappa]) \sim_{\gamma_1[f_\kappa/y]} A[y]$. So $F(\Phi) \vdash A[y]\leadsto \emptyset$.
     \item Step Case. Suppose 
       
       \noindent $\Phi \vdash \{A_1, ..., A_i,..., A_n\} \leadsto_{\kappa, \gamma_1} \{\gamma_1 A_1,...,  \gamma_1 B_1,...,  \gamma_1 B_m,..., \gamma_1 A_n\} \leadsto^*_{\gamma} \emptyset$,

\noindent  where $\kappa : \forall \underline{x}. B_1,..., B_m \Rightarrow C \in \Phi$, $C \sim_{ \gamma_1} A_i$. So we know that 

\noindent $\kappa : \forall \underline{x}. B_1[z_1],..., B_m[z_m] \Rightarrow C[f_\kappa(\underline{z})] \in F(\Phi)$ and $C[f_\kappa(\underline{z})] \sim_{\gamma_1[f_\kappa(\underline{z})/y_i]} A_i[y_i]$. Thus $F(\Phi) \vdash \{A_1[y_1], ..., A_i[y_i],..., A_n[y_n]\} \leadsto_{\kappa, \gamma_1[f_\kappa(\underline{z})/y_i]} $ 

{\scriptsize
 $\{\gamma_1[f_\kappa(\underline{z})/y_i] A_1[y_1],...,  \gamma_1[f_\kappa(\underline{z})/y_i] B_1[z_1],...,  \gamma_1[f_\kappa(\underline{z})/y_i] B_m[z_m],..., \gamma_1[f_\kappa(\underline{z})/y_i] A_n[y_n]\}$
}
\noindent $ \equiv \{\gamma_1 A_1[y_1],...,  \gamma_1 B_1[z_1],...,  \gamma_1 B_m[z_m],..., \gamma_1 A_n[y_n]\}$. By IH, 

\noindent $F(\Phi) \vdash \{\gamma_1 A_1[y_1],...,  \gamma_1 B_1[z_1],...,  \gamma_1 B_m[z_m],..., \gamma_1 A_n[y_n]\} \leadsto^* \emptyset$.

  \end{itemize}
\end{proof}
\begin{lemma}
If $F(\Phi) \vdash \{A_1[y_1],..., A_n[y_n]\} \leadsto^* \emptyset$, then $\Phi \vdash \{A_1,..., A_n\} \leadsto^* \emptyset$. 
\end{lemma}
\begin{proof}
  By induction on the length of reduction.
  \begin{itemize}
     \item Base Case. Suppose the length is one, namely, $F(\Phi) \vdash \{A[y]\} \leadsto_{\kappa, \gamma_1} \emptyset$. Thus there exists $(\kappa : \forall \underline{x} .  \Rightarrow C[f_\kappa]) \in F(\Phi)$ such that $C[f_\kappa] \sim_{\gamma_1} A[y]$.  Thus $  C \sim_{\gamma_1-[f_\kappa/y]} A$. So $\Phi \vdash A\leadsto \emptyset$.
     \item Step Case. Suppose $F(\Phi) \vdash \{A_1[y_1], ..., A_i[y_i],..., A_n[y_n]\} \leadsto_{\kappa, \gamma_1} $ 

$\{\gamma_1 A_1[y_1],...,  \gamma_1 B_1[z_1],...,  \gamma_1 B_m[z_m],..., \gamma_1 A_n[y_n]\} \leadsto^*_{\gamma} \emptyset$,

\noindent  where $\kappa : \forall \underline{x} . \forall \underline{z} . B_1[z_m],..., B_m[z_m] \Rightarrow C[f_\kappa(z_1,..., z_m)] \in F(\Phi)$,

\noindent and $C[f_\kappa(z_1,..., z_m)] \sim_{ \gamma_1} A_i[y_i]$. So we know $C \sim_{\gamma_1-[f_\kappa(\underline{z})/y_i]} A_i$. Let $\gamma = \gamma_1-[f_\kappa(\underline{z})/y_i]$. We have

\noindent $\Phi \vdash \{A_1,..., A_i,..., A_n\} \leadsto \{\gamma A_1,..., \gamma B_1,..., \gamma B_m, ..., \gamma A_n\}$

$\equiv \{\gamma_1 A_1,..., \gamma_1 B_1,..., \gamma_1 B_m, ..., \gamma_1 A_n\}$. By IH,
we know

\noindent $\Phi \vdash \{\gamma_1 A_1,..., \gamma_1 B_1,..., \gamma_1 B_m, ..., \gamma_1 A_n\} \leadsto^* \emptyset$.

  \end{itemize}
  
\end{proof}

\section{Proof of Theorem \ref{ortho:equiv}}

\begin{lemma}
  \label{tm-to-unif}
If $\Phi \vdash \{D_1,..., D_i,..., D_n\} \to_{\kappa, \gamma} \{D_1,.., \sigma E_1,..., \sigma E_m,..., D_n\}$, with $\kappa : \forall \underline{x}. \underline{E} \Rightarrow C \in \Phi$ and $C \mapsto_\sigma D_i$ for any $\gamma$, then $\Phi \vdash \{D_1,..., D_i, ..., D_n\} \leadsto_{\kappa, \gamma} $

\noindent $\{D_1,.., \sigma E_1,..., \sigma E_m,..., D_n\}$. 
\end{lemma}
\begin{proof}
  Since for $\Phi \vdash \{D_1,..., D_i,..., D_n\} \to_{\kappa, \gamma} \{D_1,.., \sigma E_1,..., \sigma E_m,..., D_n\}$, with $\kappa : \forall \underline{x}. \underline{E} \Rightarrow C \in \Phi$ and $C \mapsto_\sigma D_i$, we have $\Phi \vdash \{D_1,..., D_i, ..., D_n\} \leadsto_{\kappa, \sigma \cdot \gamma} \{\sigma D_1,.., \sigma E_1,..., \sigma E_m,..., \sigma D_n\}$. But $\mathrm{dom}(\sigma) \in \mathrm{FV}(C)$, thus we have
  
  \noindent $\Phi \vdash \{D_1,..., D_i, ..., D_n\} \leadsto_{\kappa, \gamma} \{D_1,.., \sigma E_1,..., \sigma E_m,..., D_n\}$. 

\end{proof}

\begin{lemma}\label{struct-to-unif}
  \
  
\noindent Given $\Phi$ is non-overlapping, if $\Phi \vdash \{A_1,..., A_n\}  (\hookrightarrow_{\kappa, \gamma} \cdot \to^\mu_\gamma) \{C_1, ..., C_m\}$, then $\Phi \vdash \{A_1, ..., A_n\} \leadsto^*_\gamma \{C_1, ..., C_m\}$. 
\end{lemma}
\begin{proof}
  Given $\Phi \vdash \{A_1,..., A_n\} (\hookrightarrow_{\kappa, \gamma} \cdot \to^\mu_\gamma) \{C_1, ..., C_m\}$, we know the actual reduction path must be of the form $\Phi \vdash \{A_1,..., A_n\}\hookrightarrow_{\kappa, \gamma} \{\gamma A_1,..., \gamma A_n\} \to_{\kappa, \gamma} \{\gamma A_1, ..., \gamma B_1, ..., \gamma B_n,..., \gamma A_n\} \to^\mu_{\gamma} \{C_1, ..., C_m\}$. Note that $\gamma$ is unchanged along the term-matching reduction. The $\to$ following right after $\hookrightarrow$ can not use a different rule other than $\kappa$, it would mean $\gamma A_i \equiv \gamma C$ with $\kappa : \forall \underline{x}.\underline{B} \Rightarrow C \in \Phi$ and $A_i \equiv \sigma B$ with $\kappa' : \forall \underline{x}. \underline{D} \Rightarrow B \in \Phi$. This implies $\gamma C \equiv \gamma \sigma B$, contradicting
the non-overlapping restriction. Thus we have $\Phi \vdash \{A_1,..., A_n\}\leadsto_{\kappa,  \gamma} \{\gamma A_1, ..., \gamma B_1, ..., \gamma B_n,..., \gamma A_n\}$. By Lemma \ref{tm-to-unif},
we have $\Phi \vdash \{A_1,..., A_n\}\leadsto_{\kappa,  \gamma} \{\gamma A_1, ..., \gamma B_1, ..., \gamma B_n,..., \gamma A_n\} \leadsto^*_\gamma \{C_1, ..., C_m\}$

\end{proof}

\begin{lemma}
  Given $\Phi$ is non-overlapping, if $\Phi \vdash \{A_1,..., A_n\} (\to^\mu \cdot \hookrightarrow^1)^*_\gamma  \{C_1, ..., C_m\}$ with $\{C_1, ..., C_m\}$ in $\to^\mu \cdot \hookrightarrow^1$-normal form, then $\Phi \vdash \{A_1, ..., A_n\} \leadsto^*_\gamma \{C_1, ..., C_m\}$ with $\{C_1, ..., C_m\}$ in $\leadsto$-normal form. 
\end{lemma}
\begin{proof}
  Since $\Phi \vdash \{A_1, ..., A_n\}  (\to^\mu \cdot \hookrightarrow^1)^*_\gamma \{C_1, ..., C_m\}$, this means the reduction path must be of the form $\Phi \vdash \{A_1,..., A_n\} \to^\mu \cdot \hookrightarrow^1 \cdot \to^\mu \cdot \hookrightarrow^1 ... \to^\mu \cdot \hookrightarrow^1 \cdot \to^\mu  \{C_1, ..., C_m\}$. Thus $\Phi \vdash \{A_1,..., A_n\} \to^\mu \cdot (\hookrightarrow^1 \cdot \to^\mu) \cdot (\hookrightarrow^1 ... \to^\mu) \cdot (\hookrightarrow^1 \cdot  \to^\mu)  \{C_1, ..., C_m\}$. By Lemma \ref{tm-to-unif} and Lemma \ref{struct-to-unif}, we have $\Phi \vdash \{A_1, ..., A_n\} \leadsto^*_\gamma \{C_1, ..., C_m\}$ with $\{C_1, ..., C_m\}$ in $\leadsto$-normal form.
  
\end{proof}

\begin{lemma}
  Given $\Phi$ is a non-overlapping, if $\Phi \vdash \{A_1, ..., A_n\} \leadsto^*_\gamma \{C_1, ..., C_m\}$ with $\{C_1, ..., C_m\}$ in $\leadsto$-normal form , then $\Phi \vdash \{A_1, ..., A_n\} (\to^\mu \cdot \hookrightarrow^1)^*_\gamma \{C_1, ..., C_m\}$ with $\{C_1, ..., C_m\}$ in $\to^\mu \cdot \hookrightarrow^1$-normal form. 
\end{lemma}
\begin{proof}
  By induction on the length of $\leadsto^*_\gamma$. 

\begin{itemize}
\item Base Case: $\Phi \vdash \{A_1,...,A_i, ..., A_n\} \leadsto_{\kappa, \gamma} \{\gamma A_1,...,\gamma B_1, ..., \gamma B_m ..., \gamma A_n \}$ with $\kappa : \forall \underline{x}. \ \underline{B}\Rightarrow C \in \Phi$, $C \sim_\gamma A_i$ and $\{\gamma A_1,...,\gamma B_1, ..., \gamma B_m ..., \gamma A_n \}$ in $\leadsto$-normal form . We have $\Phi \vdash \{A_1,...,A_i, ..., A_n\} \hookrightarrow_{\kappa, \gamma} \{\gamma A_1,...,\gamma A_i, ..., \gamma A_n \} \to_\kappa \{\gamma A_1,...,\gamma B_1, ..., \gamma B_m ..., \gamma A_n\}$ with $\{\gamma A_1,...,\gamma B_1, ..., \gamma B_m ..., \gamma A_n\}$ in $\to^\mu\cdot \hookrightarrow$-normal form. Note that there can not be another
$\kappa' : \forall \underline{x}. \underline{B} \Rightarrow C' \in \Phi$ such that $\sigma C' \equiv A_i$, since this would means $\gamma C \equiv \gamma A_i \equiv \gamma \sigma C'$, violating the non-overlapping requirement.

\item Step Case: $\Phi \vdash \{A_1, ..., A_i, ..., A_n\} \leadsto_{\kappa, \gamma} \{\gamma A_1,..., \gamma B_1, ..., \gamma B_l, ..., \gamma A_n\} \leadsto^*_{\gamma'} \{C_1, ..., C_m\}$ with $\kappa : \forall \underline{x}. B_1,..., B_l \Rightarrow C \in \Phi$ and $C \sim_\gamma A_i$. 

\noindent We have $\Phi \vdash \{A_1, ..., A_i, ..., A_n\} \hookrightarrow_{\kappa, \gamma} \{\gamma A_1, ..., \gamma A_i, ..., \gamma A_n\}\to $

\noindent $\{\gamma A_1,..., \gamma B_1, ..., \gamma B_m, ..., \gamma A_n\}$. By the non-overlapping requirement, there can not be another $\kappa' : \forall \underline{x}. \underline{D} \Rightarrow C' \in \Phi$ such that $\sigma C' \equiv A_i$. 

\noindent By IH, we know $\Phi \vdash \{\gamma A_1,..., \gamma B_1, ..., \gamma B_m, ..., \gamma A_n\} (\to^\mu \cdot \hookrightarrow)_{\gamma'}^* \{C_1, ..., C_m\}$. Thus we conclude that $\Phi \vdash \{A_1, ..., A_i, ..., A_n\} (\hookrightarrow \cdot \to)^*_{\gamma'} \{C_1, ..., C_m\}$. 

\end{itemize}

\end{proof}

\section{Proof of Theorem \ref{prod-non-overlap}}
We assume a non-overlapping and productive program $\Phi$ in this section.

\begin{lemma}
  If $\Phi \vdash \{A_1,..., A_n\} \leadsto \{B_1,..., B_m\}$, then $\Phi \vdash \{A_1,..., A_n\} (\to^\mu \cdot \hookrightarrow^1)^* \{C_1,..., C_l\}$ and $\Phi \vdash \{B_1,..., B_m\} \to^* \{C_1,..., C_l\}$.
\end{lemma}
\begin{proof}
  Suppose $\Phi \vdash \{A_1,..., A_n\} \leadsto_{\kappa, \gamma} \{\gamma A_1,...,\gamma E_1, ..., \gamma E_l,..., \gamma A_n\}$, with $\kappa : \underline{E} \Rightarrow D \in \Phi$ and $D \sim_\gamma A_i$. Suppose $D \not \mapsto_\gamma A_i$. In this case, we have $\Phi \vdash \{A_1,..., A_n\} \hookrightarrow_{\kappa, \gamma} \cdot \to_{\kappa, \gamma} \{\gamma A_1,...,\gamma E_1, ..., \gamma E_q,..., \gamma A_n\} \to^\mu_\gamma \{C_1,..., C_l\}$. Suppose $D \mapsto_\gamma A_i$, we have $\Phi \vdash \{A_1,..., A_n\} \to_{\kappa, \gamma} \{\gamma A_1,...,\gamma E_1, ..., \gamma E_q,..., \gamma A_n\} \to^\mu_\gamma \{C_1,..., C_l\}$.
\end{proof}
\begin{lemma}\label{I}
  If $\Phi \vdash \{A_1,..., A_n\} \hookrightarrow_{\kappa,\gamma} \{\gamma A_1,..., \gamma A_n\} \to^\mu_\gamma \{B_1,..., B_m\}$, then $\Phi \vdash \{A_1,..., A_n\} \leadsto^*_\gamma \{B_1,..., B_m\}$. \end{lemma}
\begin{proof}
Suppose $\Phi \vdash \{A_1,..., A_n\} \hookrightarrow_{\kappa,\gamma} \{\gamma A_1,..., \gamma A_n\} \to^\mu_\gamma \{B_1,..., B_m\}$, we have $\Phi \vdash \{A_1,..., A_n\} \hookrightarrow_{\kappa,\gamma} \{\gamma A_1,..., \gamma A_n\} \to_\kappa \{\gamma A_1,..., \gamma C_1,..., \gamma C_l. \gamma A_n\}\to^\mu_\gamma \{B_1,..., B_m\}$ with $\kappa : \underline{C} \Rightarrow D \in \Phi$ and $D \sim_\gamma A_i$. Thus we have $\Phi \vdash \{A_1,..., A_n\} \leadsto_{\kappa,\gamma}  \{\gamma A_1,..., \gamma C_1,..., \gamma C_l,..., \gamma A_n\}$. By Lemma \ref{tm-to-unif}, we have $\Phi \vdash \{A_1,..., A_n\} \leadsto_{\kappa,\gamma}  \{\gamma A_1,..., \gamma C_1,..., \gamma C_l, ...., \gamma A_n\} \leadsto^*_\gamma \{B_1,..., B_m\}$.
\end{proof}
\begin{lemma}
If $\Phi \vdash \{A_1,..., A_n\} (\to^\mu \cdot \hookrightarrow^1)^*_\gamma \{B_1,..., B_m\}$, then $\Phi \vdash \{A_1,..., A_n\} \leadsto^*_\gamma \{B_1,..., B_m\}$. 
\end{lemma}
\begin{proof}
  By Lemma \ref{I}.
\end{proof}
\end{document}